\documentclass[conference]{IEEEtran}
\IEEEoverridecommandlockouts
\newcommand{\customfootnoterule}{\hspace{-\parindent}\rule{.618\columnwidth}{0.5pt}\newline\vspace{3pt}}
\usepackage{xcolor}
\usepackage{balance}
\usepackage{comment}
\usepackage[font=footnotesize]{caption}
\newcommand{\spacetune}[1]{#1}
\usepackage{amsmath}
\usepackage{graphicx}
\usepackage{amsmath}
\usepackage[byname]{smartref}
\usepackage{txfonts}
\usepackage{tocloft}
\usepackage{afterpage}
\usepackage{amssymb}
\usepackage{stmaryrd}
\usepackage{setspace}
\usepackage{attrib}
\usepackage{pgfplots}
\pgfplotsset{compat=1.18}
\usepackage{subcaption}
\usepackage{mdwlist}
\usepackage{listings}
\usepackage{tikz}
\usetikzlibrary{matrix,decorations.pathreplacing}
\usepackage[framemethod=TikZ]{mdframed}
\usepackage[utf8x]{inputenc}
\usepackage[T1]{fontenc}
\newcommand{\hset}[1]{ \left\{ \! \left| \; #1 \; \right| \! \right\} }

\newcommand{\restrict}[2]{\left.#1\right|_{#2}}

\DeclareMathOperator{\supp}{supp}

\newcommand{\R}[1][]{\mathcal{R}_{#1}}

\newcommand{\pseudo}[1]{\widetilde{#1\vphantom{f}}}

\def\ind[#1]{1_{#1}}

\def\hsetover[#1]{\mathbb{Z}^{#1}}

\newcommand{\mult}{\times}
\newcommand{\smallmult}{\times}
\makeatletter
{\clearemptydoublepage
 \begin{center}
  \section*{Acknowledgements}
 \end{center}
 \begingroup
}{\newpage\endgroup}
{\clearemptydoublepage 
 \begin{center}
  \section*{Dedication}
 \end{center}
 \begingroup
}{\newpage\endgroup}
{\pagestyle{plain}\pagenumbering{roman}}%
{\pagenumbering{arabic}}
\newtheorem{theorem}{Theorem}[section]

\newenvironment{proof}[1][Proof]{\begin{trivlist}
\item[\hskip \labelsep {\bfseries #1}]}{\end{trivlist}}
\newenvironment{definition}[1][Definition]{\begin{trivlist}
\item[\hskip \labelsep {\bfseries #1}]}{\end{trivlist}}

\newcommand{\qed}{\nobreak \ifvmode \relax \else
      \ifdim\lastskip<1.5em \hskip-\lastskip
      \hskip1.5em plus0em minus0.5em \fi \nobreak\ensuremath{\square}\fi}
\makeatother
\title{Hybrid Intervals and Symbolic Block Matrices}
\author{
\IEEEauthorblockN{Mike Ghesquiere$^\dagger$
   \thanks{\customfootnoterule $^\dagger$%
            Present address: Joy Life Inc, 407 E Duarte Rd $\#$E, Arcadia CA 91006, USA}}
\IEEEauthorblockA{Computer Science Department\\
University of Western Ontario, Canada\\
\texttt{mike.ghes@gmail.com} }
\and
\IEEEauthorblockN{Stephen M. Watt}
\IEEEauthorblockA{Cheriton School of Computer Science \\
University of Waterloo, Canada\\
\texttt{smwatt@uwaterloo.ca} }
}
\date{}
\begin{document}
\maketitle
\begin{abstract}
Structured matrices with symbolic  sizes appear frequently in the literature, especially in the description of algorithms for linear algebra.
Recent work has treated these symbolic structured matrices themselves as computational objects, showing how to add matrices with blocks of different symbolic sizes in a general way while avoiding a combinatorial explosion of cases.
The present article introduces the concept of hybrid intervals, in which points may have negative multiplicity.  
Various operations on hybrid intervals have compact and elegant formulations that do not require cases to handle different orders of the end points. 
This makes them useful to represent symbolic block matrix structures and to express arithmetic on symbolic block matrices compactly.   
We use these ideas to formulate symbolic block matrix addition and multiplication in a compact and uniform way.
\end{abstract}
%
%
\section{Introduction}
\label{sec:Introduction}
Block and other structured matrices appear throughout the mathematical and computer science literature.
They occur when systems have direct sum decompositions, when higher order systems are represented as Kronecker products, and when recursive methods are used.
They figure prominently in efficient algorithms for linear algebra, such as Strassen's $n^{\log_27}$ matrix multiplication and $LU$ decomposition.  
Block matrices have considerable practical applications as well.
For example, when multiplying large matrices, block algorithms can be used to improve cache complexity~\cite{lam1991cache}.
Additionally, in some cases, when a sub-matrices are known to have useful properties, many optimizations can arise.
For example to invert a block diagonal matrix, one can invert each block individually.

In the literature, the parts of structured matrices are often given with symbolic size, for example blocks of size $n\times m$ or diagonal bands of width $k$.  
These matrices are usually taken as inputs to algorithms that work with specific instances with particular values for the size parameters. 
It is well-understood how to work with such matrices as algebraic values when the size parameters are fixed.  

What is less well understood is how to do algebraic computation on block matrices as whole symbolic objects,
that is when the size parameters are symbolic. This remains an active area of research.   
As with expressions involving any piecewise functions,
one of the principal problems is dealing with the multitude of cases that arise in performing algebraic operations on matrices when the relationships among the symbolic size parameters are not fixed.   
This is illustrated in Figure~\ref{fig:MatAdditionPermutations}.

Suppose we have an expression involving $n$ binary operations on symbolic block matrices.  
Various approaches have been taken to address the proliferation of cases:

The most obvious approach is to enumerate all cases, giving a symbolic expression for the algebraic result in each case.  The problem with this approach is that it leads to a number of cases exponential in $n$.

Another approach is to create a single expression covering all cases, using multiplicative support functions to zero out the parts that do not apply in specific cases.  
This is the approach taken in~\cite{sexton2008abstract} and~\cite{sexton2009computing}, with $\sigma_{ij}$ being support functions that take on the values $0$ or $1$.  
While an algebra defined on the support functions allows some simplification, in general the resulting expressions can be of size exponential in $n$.

For operations with inverses, another approach is to presume a specific ordering of the symbolic size parameters and use generalized support functions to add or remove components as required.   
This is essentially what is behind the convention of orientation of integral and sum limits, so one can have identities such as $\int_a^b = \int_a^c + \int_c^b$, regardless of the ordering of $a$, $b$ and $c$.
This approach is taken in~\cite{sexton2009reasoning}, with $\xi_{ijk}$ being generalized multiplicative support functions that take on the values $0$, $1$ or $-1$.  
This leads to compact expressions, avoiding exponential growth, but can be applied only for operations having total inverses.
Aside from being unable to handle non-invertible operations, another problem with this method is that
it performs extra calculations to compute values and then use inverses to cancel them out.  

A more sophisticated approach is to realize that the use of inverse operations is actually a proxy for reversing the decision to include certain operands in expressions.   
It is possible to instead make the choice of operand inclusion or exclusion explicit.  
For functions that are associative and commutative, one can instead have a symbolic approach to gathering operands.  
This is the approach taken in~\cite{carette2010} and further studied in~\cite{mikegmasters}. 
Generalized multisets allowing negative multiplicities are used to collect arguments without concern for the order of inclusion or exclusion.  
These generalized multisets are known as ``hybrid sets.'' 
Previous work has shown that they allow efficient general composition of piece-wise symbolic functions, including matrices with symbolically defined regions.  In the present article we show how these ideas take a particularly elegant form using the notion of ``hybrid intervals,'' which are intervals generalized as hybrid sets.

\spacetune{\pagebreak}
The main contributions of this article are:
\begin{itemize}
    \item the concept of hybrid intervals, capturing the notion of inclusion and exclusion, 
    \item useful results about their properties,
    \item the use of hybrid intervals to reformulate symbolic block matrix addition elegantly,
    \item the use of hybrid intervals to express symbolic block matrix multiplication in a similar fashion.
\end{itemize}

The remainder of this article is organized as follows:
Section~\ref{sec:HybridSets} provides the required background on hybrid sets and hybrid functions.
Section~\ref{sec:HybridIntervals} introduces the notion of hybrid intervals and gives some of their properties.
Section~\ref{sec:VectorAddition} gives as an example using hybrid intervals of index sets for the addition of vectors whose structure is symbolically parameterized.  This is a familiar example from previous work, shown here to illustrate the present approach.
Section~\ref{sec:HigherDimensionIntervals} generalizes the idea of hybrid intervals to higher dimension.
Section~\ref{sec:MatrixAddition} shows how two-dimensional hybrid intervals of index sets may be used to formulate the addition of matrices with symbolic structure.
Section~\ref{sec:MatrixMultiplication} applies these ideas to matrix multiplication.
Some conclusions are given in Section~\ref{sec:Conclusions}.

\begin{figure}[t]
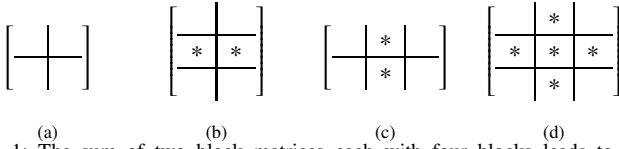
 
	\begin{subfigure}[b]{0.24\columnwidth}
		\begin{equation*}
			\left[ \begin{array}{c|c} \; & \; \\ \hline & \\ \end{array}\right]
		\end{equation*}
		\caption{}
	\end{subfigure}
	\begin{subfigure}[b]{0.24\columnwidth}
		\begin{equation*}
			\left[ \begin{array}{c|c} \; & \; \\ \hline * & * \\ \hline & \\ \end{array}\right]
		\end{equation*}\vspace{-4pt}
		\caption{}
	\end{subfigure}
	\begin{subfigure}[b]{0.24\columnwidth}
		\begin{equation*}
			\left[ \begin{array}{c|c|c} \; & * & \;  \\ \hline \; & * &\; \\ \end{array}\right]
		\end{equation*}
		\caption{}
	\end{subfigure}
	\begin{subfigure}[b]{0.24\columnwidth}
		\begin{equation*}
			\left[ \begin{array}{c|c|c} \; & * & \; \\ \hline * & * & * \\ \hline \; & * & \; \\ \end{array}\right]
		\end{equation*}\vspace{-4pt}
		\caption{}
	\end{subfigure}
	\caption[Possible block overlaps of $2 \times 2$ block matrices.] {
		The sum of two block matrices each with four blocks leads to 9 possible cases.
		When blocks are the same size, the sum will also be a $2 \times 2$ block matrix (a).
		Otherwise, a $2 \times 3$ (b) \emph{(two cases)}, $3 \times 2$ (c) \emph{(two cases)} 
		or $3 \times 3$ (d) \emph{(four cases)} block matrix could arise.
		The starred blocks may sample from different blocks depending on the relative size of operand blocks.}
	\label{fig:MatAdditionPermutations}
\end{figure}

%
%
\section{Hybrid Sets and Hybrid Functions}
\label{sec:HybridSets}
We summarize here the main concepts from~\cite{carette2010} used in this article. To begin, we note that the usual notion of a set $A$ of elements from a universe $U$ can be defined as a characteristic function $U \rightarrow \{0,1\}$ that gives $1$ for elements of $A$ and $0$ otherwise.   
A multiset may be viewed as a function $U \rightarrow \mathbb N_0$ that gives the multiplicity of an element.  
With this in mind, we make the following definition that provides negative multiplicities:
\begin{definition}
A \textbf{hybrid set} over a universe set $U$ is an integer-valued function $U\rightarrow \mathbb Z$
giving a \textbf{multiplicity} of each element.
\end{definition}
In what follows, some types of objects are generalized with hybrid sets.  
When necessary to make the distinction, we call the usual objects ``traditional'', \textit{e.g.} traditional sets.

If $H$ is a hybrid set, we write $H(x)$ for the multiplicity of $x$ in $H$.  We say $x \in H$ if $H(x) \ne 0$ and define the support $\supp H$ to be the traditional set $\{ x | x \in H \}$.  We use the notation $\hset{a^i, b^j, ...}$ to denote a hybrid set containing $a$ with multiplicity $i$, $b$ with multiplicity $j$, \textit{etc}.  
The empty hybrid set, denoted $\varnothing$, is one for which all elements have multiplicity zero.
The set operations $\cup$, $\cap$ and $\backslash$ may be defined on hybrid sets, but more useful are the element-wise combinations of multiplicities:
\begin{definition}
	For any two hybrid sets $A$ and $B$ over a common universe $U$, 
	we define the \textbf{operations} $\boldsymbol\oplus, \boldsymbol\ominus, \boldsymbol\otimes : \mathbb{Z}^U \times \mathbb{Z}^U \to \mathbb{Z}^U$ 
	such that for all $x \in U$:
	\begin{align*}
		(A \oplus B)(x) 	&= A(x) + B(x)  &
		(A \ominus B)(x) 	&= A(x) - B(x) \\
		(A \otimes B)(x) 	&= A(x) \times B(x) &
             \ominus A &= \varnothing \ominus A
	\end{align*}
	and for $c \in \mathbb{Z}$,
	\( (cA)(x) = c \times A(x) \).
\end{definition}
We shall use the following definitions:
\begin{definition}
	We say $\boldsymbol{A}$ \textbf{and} $\boldsymbol{B}$ \textbf{are disjoint} when $A \otimes B = \varnothing$.
\end{definition}
\begin{definition}
	A \textbf{generalized partition $\boldsymbol{P}$ of a hybrid set $\boldsymbol{H(x)}$} is a family of hybrid sets
	${P=\{P_i \}_{i=1}^n}$ such that:
	\begin{equation*}
		H = P_1 \oplus P_2 \oplus \ldots \oplus P_n
	\end{equation*}
	We say that \textbf{$\boldsymbol{P}$ is a strict partition of $\boldsymbol{H}$} if 
	$P_i$ and $P_j$ are disjoint when $i \neq j$.
\end{definition}
Clearly any traditional set may be considered a hybrid set. For the other direction, we have
\begin{definition}
	Given a hybrid set $H$ over universe $U$, 
	if for all $x \in U$, $H(x) \in \{0,1\}$, then we say that \textbf{$\boldsymbol{H(x)}$ is reducible}.
	If $H$ is reducible then we denote the \textbf{reduction of $\boldsymbol{H}$} by $\mathcal{R}(H)$ 
	as the (non-hybrid) set over $U$ with the same membership.  
\end{definition}
In what follows, we require the notion hybrid functions and some of their properties. We define these in terms of the graph of the function as a hybrid set.
\begin{definition}
	For two traditional sets $S$ and $T$, a hybrid set over their Cartesian product $S \times T$ is called a 
	\textbf{hybrid (binary) relation between $\boldsymbol{S}$ and $\boldsymbol{T}$}.
	A \textbf{hybrid function from $\boldsymbol{S}$ to $\boldsymbol{T}$} is 
	a hybrid relation $h$ between $S$ and $T$ such that $(x,y) \in h$ and $(x,z) \in h$ implies $y=z$.
\end{definition}
Given a hybrid set $H$ over $U$ and a function $f:B \to S$ with $B \subseteq U$ and $S$ a set,
we denote by $f^H$ the hybrid function from $B$ to $S$ defined by:
\begin{equation*}
	f^H := \bigoplus_{x \in B} H(x) \hset{ (x, f(x) )^1 }
\end{equation*}
\begin{definition}
	If $H$ is a reducible hybrid set, then \textbf{$\boldsymbol{f^H}$ is a reducible hybrid function} and
 we extend $\mathcal{R}$ by
	\begin{equation*}
		\mathcal{R}(f^H)(x) = \restrict{f}{\mathcal{R}(H)}(x)
	\end{equation*}
\end{definition}
With these definitions, we have for functions $f, g: U \to S$,
\( f^A \oplus f^B = f^{A \oplus B} \).
Also, for hybrid functions $f^A$ and $g^B$, $f^A \oplus g^B$ is a hybrid function
if and only if for all $x \in \supp (A \otimes B)$, we have $f(x) = g(x)$. In this case
we say that \textbf{$\boldsymbol{f^A}$ and $\boldsymbol{g^B}$ are compatible}.

One of the main points in what follows is that we will have expressions where collections of arguments are formed as hybrid sets before it is known whether the subexpressions are defined over the required domains.  Sub-expressions of undefined value are completely acceptable if their multiplicity is zero in the argument collection as a whole. For these subexpressions to be well-formed, we employ a notion of ``pseudo-functions''.
\begin{definition}
	We define a pseudo-function $\pseudo{f\;}^A$ as:
	\begin{equation}
		\label{eqn:pseudofunc}
 		\pseudo{f\;}^A = \bigoplus_{x \in B} A(x) \hset{(x,f)^1}
	\end{equation}
\end{definition}

The difference between hybrid functions and pseudo-functions is that $(x, f(x))$ is replaced with the ``unevaluated'' $(x,f)$.
This formally makes $\pseudo{f\;}^A$ a hybrid relation over $U \times (U \to S)$ as opposed to 
a hybrid function over $U \times S$.
To evaluate $\pseudo{f\;}^A$ we map back to $f^A$.
This mapping between $(x,f(x))$ and $(x,f)$ is natural and we will perform it as necessary, usually without comment.

%
%

\section{Hybrid Intervals}
\label{sec:HybridIntervals}

\begin{definition}
	Given a totally ordered set $(X, \leq)$ \emph{(and with an implied strict ordering $<$)}, 
	for any $a,b \in X$, an \textbf{interval between $\boldsymbol{a}$ and $\boldsymbol{b}$} 
	is the set of elements in $X$ between $a$ and $b$, as follows:
	\begin{align*} 
		{[a,b]}_X &=  \{ x \in X \;|\; a \leq x \leq b \} &
		{[a,b)}_X &=  \{ x \in X \;|\; a \leq x < b \} \\
		{(a,b]}_X &=  \{ x \in X \;|\; a < x \leq b \} &
		{(a,b)}_X &=  \{ x \in X \;|\; a < x < b \}
	\end{align*}
	When context makes $X$ obvious or the choice of $X$ is irrelevant, the subscript may be omitted.
\end{definition}

It should be noted that when $b$ is less than $a$, $[a,b]$ is the empty set. 
In terms of idempotency, the bounds determine whether or not an interval is empty.
$[a,a]$ which contains $a$ and all points equivalent to $a$ while $(a,a)$, $(a,a]$, and $[a,a)$ are all empty sets.
As intervals are simply sets, they can naturally be interpreted as hybrid sets.
If $a \leq b \leq c$, for intervals then we have $[a,b) \oplus [b,c) = [a,c)$.
In this case, $\oplus$ seems to behave like concatenation but this is not always true.
Considering all possible relative orders of $a$, $b$ and $c$ gives
\begin{equation*}
	[a,b) \oplus [b,c) =
	\begin{cases}
		\; [a,c) & a \leq b \leq c \\
		\; [a,b) & a \leq c \leq b \lor c \leq a \leq b \\
		\; [b,c) & b \leq a \leq c \lor b \leq c \leq a\\
		\; \varnothing & c \leq b \leq a. \\
	\end{cases}
\end{equation*}

One could alternatively write $[a,b)\oplus [b,c) = [\; \min(a,b),\max(b,c) \;)$ but this simply sweeps the problem 
under the rug.
When working with intervals, a case-based approach to consider relative ordering of endpoints easily becomes 
quite cumbersome.
Previously, the $\xi$ function was introduced in \cite{sexton2008abstract} to solve this problem.
Although it solves the problem of cases, it quickly leads to heavy notation.
Instead we introduce hybrid intervals which are considerably more readable.
It should be noted that the definitions are equivalent; $\xi(i,y,z)$ and $[\![y,z)\!)$ can be used interchangeably.

\begin{definition}
	We define \textbf{hybrid intervals} with $a,b\in X$, where $X$ is a totally ordered set, 
	using hybrid set point-wise subtraction as follows:
	\begin{align*}
			{[\![ a,b )\!)} &= [a,b) \ominus [b,a) &
			{(\!( a,b ]\!]} &= (a,b] \ominus (b,a] \\
			{[\![ a,b ]\!]} &= [a,b] \ominus (b,a) &
			{(\!( a,b )\!)} &= (a,b) \ominus [b,a]
	\end{align*}
\end{definition}

For any choice of \emph{distinct} $a$ and $b$, exactly one term will be empty; there can be no ``mixed'' multiplicities from a single hybrid interval.
Unlike traditional intervals where $[a,b)$ would be empty if $b < a$,  
the hybrid interval $[\![a,b)\!)$ will have elements with negative multiplicity.
Several results follow immediately from this definition.

\begin{theorem} For all $a,b,c$, 
	\begin{align*}
		{[\![a,b)\!)} &= \ominus\, [\![b,a)\!) &
		{(\!(a,b]\!]} &= \ominus\, (\!(b,a]\!] \\
		{[\![a,b]\!]} &= \ominus\, (\!(a,b)\!) &
		{(\!(a,b)\!)} &= \ominus\, [\![a,b]\!]
	\end{align*}
\end{theorem}

\begin{proof}
	All the identities can all be shown in the same fashion:
	\begin{align*}
		[\![a,b)\!) = [a,b) \ominus [b,a) = \ominus \big( [b,a) \ominus [a,b) \big) = \ominus [\![b,a)\!) \\ 
		(\!(a,b]\!] = (a,b] \ominus (b,a] = \ominus \big( (b,a] \ominus (a,b] \big) = \ominus (\!(b,a]\!] \\ 
		[\![a,b]\!] = [a,b] \ominus (b,a) = \ominus \big( (b,a) \ominus [a,b] \big) = \ominus (\!(b,a)\!)
	\end{align*}
	Since $[\![a,b]\!] = \ominus (\!(b,a)\!)$ we also have $(\!(a,b)\!) = \ominus [\![b,a]\!]$. 
	\hfill $\qed$
\end{proof}

We note here how hybrid intervals behave when $a=b$.
Like their traditional analogues, the hybrid intervals $[\![ a,a )\!)$ and $(\!( a,a ]\!]$ are still both empty sets.
The interval $[\![a,a]\!]$ still contains points equivalent to $a$ (with multiplicity 1).
However, unlike traditional intervals $(\!(a,a)\!)$ is \emph{not} empty but rather, $(\!(a,a)\!) = \ominus [\![a,a]\!]$ and so contains all points equivalent to $a$ but with a multiplicity of $-1$.
The advantage of using hybrid intervals is that now $\oplus$ does behave like concatenation.

\begin{theorem}
	For all $a,b,c$ (regardless of relative ordering),
	\begin{equation*}
		[\![ a,b )\!) \oplus [\![ b,c )\!) = [\![ a,c )\!)
	\end{equation*}
\end{theorem}
\begin{proof}
	Following from definitions we have:
	\begin{align*}
		[\![a,b)\!) \oplus [\![ b,c )\!)
		& = \bigg ( [a,b) \ominus [b,a)  \bigg ) \oplus \bigg ( [b,c) \ominus [c,b) \bigg )\\
		& = \bigg ( [a,b) \oplus  [b,c) \bigg ) \ominus \bigg ( [c,b) \oplus  [b,a) \bigg )
	\end{align*}
		Case 1: $a \leq c$. We have $[c,a) = \varnothing$ and so $[\![a,c)\!) = [a,c)$. 
		\begin{description}
			\item[Case 1.a: $a \leq b \leq c$]~\\ then $[c,b) = [b,a) = \varnothing$ and $[a,b) \oplus [b,c) = [a,c)$
			\item[Case 1.b: $b \leq a \leq c$]~\\then $[b,c) \ominus [b,a) = [b,a) \oplus [a,c) \ominus [b,a) = [a,c)$
			\item[Case 1.c: $a \leq c \leq b$]~\\then $[a,b) \ominus [c,b) = ([a,c) \oplus [c,b)) \ominus [c,b) = [a,c)$
		\end{description}
		Case 2: $c < a$. We have $[a,c) = \varnothing$ and so $[\![a,c)\!) = \ominus [c,a)$. 
		\begin{description}
			\item[Case 2.a: $c \leq b \leq a$]~\\ 
				then $[a,b) = [b,c) = \varnothing$ and $\ominus [c,b) \ominus [b,a) = \ominus [c,a)$
			\item[Case 2.b: $b \leq c \leq a$]~\\ 
				then $\ominus [b,a) \oplus [b,c) = \ominus ([b,c) \oplus [c,a)) \oplus [b,c) = \ominus[c,a)$
			\item[Case 2.c: $c \leq a \leq b$]~\\ 
				then $\ominus [c,b) \oplus [a,b) = \ominus ([c,a) \oplus [a,b)) \oplus [a,b) = \ominus[c,a)$
		\end{description}
            \hfill $\qed$
\end{proof}

This sort of reasoning is routine but a constant annoyance when dealing with intervals 
and is exactly the reason we want to be working with hybrid intervals.
Now that the above work is done, we can use hybrid intervals 
and not concern ourselves with the relative ordering of points.
Many similar formulations such as $[\![ a,b ]\!] \oplus (\!( b,c )\!) = [\![a,c)\!)$ or $(\!(a,b)\!) \oplus [\![b,c)\!) = (\!(a,c)\!)$ 
are also valid for any ordering of $a,b,c$ by the same sort of argument.

%
%

\section{Vector Addition}
\label{sec:VectorAddition}
Addition for partitioned vectors and $2 \times 2$ matrices using hybrid functions has already been considered in \cite{carette2010} and~\cite{sexton2008abstract}.  
Here we show how the same may be accomplished using hybrid intervals.
We take a familiar  example from~\cite{sexton2008abstract} and adapt it to use this formulation.
This shows  hybrid intervals in use and serves as preparation for the addition and multiplication of 
symbolic block matrices using these concepts.

We consider the addition of the $n$-component vectors $U$ and $V$, each consisting of two parts with indices $[1,k]$ and $(k,n]$ for $U$ and $[1,\ell]$ and $(\ell, n]$ for $V$.
Over each interval, taking the value of different functions, as in:
\begin{align*}
	U &= [ u_1, u_2, \ldots, u_{k}, u'_1, u'_2, \ldots, u_{n-k} ] \\
	V &= [ v_1, v_2, \ldots, v_{\ell}, v'_1, v'_2, \ldots, v_{n-\ell} ].
\end{align*}
Using intervals, these vectors can be represented by hybrid functions over their indices,
for example
\begin{align*}
	U &= (i \mapsto u_i)^{[\![1, k]\!]} \oplus (i \mapsto u'_{i-k})^{(\!(k,n]\!]} \\
	V &= (i \mapsto v_i)^{[\![1, \ell]\!]} \oplus (i \mapsto v'_{i-\ell})^{(\!(\ell,n]\!]}.
\end{align*}
For clarity and succinctness we use $(u_i)$ instead of $(i \mapsto u_i)$:
\begin{align*}
	U &= (u_i)^{[\![1, k]\!]} \oplus (u'_{i-k})^{(\!(k,n]\!]} \\
	V &= (v_i)^{[\![1, \ell]\!]} \oplus v'_{i-\ell})^{(\!(\ell,n]\!]}.
\end{align*}
To add $U$ and $V$, we have
\begin{align*}
	U + V
	&= \left( (u_i)^{[\![1, k]\!]} \oplus (u'_{i-k})^{(\!(k,n]\!]} \right) 
		+
		\left( (v_i)^{[\![1, \ell]\!]} \oplus (v'_{i-\ell})^{(\!(\ell,n]\!]} \right) \\
	&= \left( (u_i)^{[\![1, k]\!]} \oplus (u'_{i-k})^{(\!(k,\ell]\!]} \oplus (u'_{i-k})^{(\!(\ell,n]\!]} \right) \\
		&+
		\left( (v_i)^{[\![1, k]\!]} \oplus (v_i)^{(\!(k, \ell]\!]} \oplus (v'_{i-\ell})^{(\!(\ell,n]\!]} \right) \\
	&= \R[+] \left( (u_i + v_i)^{[\![1, k]\!]} 
		\oplus (u'_{i-k} + v_i)^{(\!(k,\ell]\!]} 
		\oplus (u'_{i-k}+v'_{i-\ell})^{(\!(\ell,n]\!]} \right).
\end{align*}

The choice to partition $[\![1,n]\!]$ as $[\![1,k]\!] \oplus (\!(k,\ell]\!] \oplus (\!(\ell, n]\!]$ is only one possible refinement.
We can just as easily use $[\![1,\ell]\!] \oplus (\!(\ell, k]\!] \oplus (\!(k, n]\!]$ to get the equivalent expression:
\begin{equation*}
	U + V = \R[+] \left( (u_i + v_i)^{[\![1, \ell]\!]} 
		\oplus (u_{i} + v'_{i-\ell})^{(\!(\ell,k]\!]} 
		\oplus (u'_{i-k}+v'_{i-\ell})^{(\!(k,n]\!]} \right)
\end{equation*}

We must be careful while evaluating these expressions to not forget that $(u'_{i-k} + v_i)$ 
is actually shorthand for the function:
\begin{equation*}
	(u'_{i-k} + v_i) = (i \mapsto u'_{i-k}) + (i \mapsto v_i) = (i \mapsto u'_{i-k} + v_i).
\end{equation*}
As a function, it may not be evaluable over the entire range implied in a given term.
Using pseudo-functions easily solves this.

For example, consider the concrete example where $n=5$, $k=4$ and $\ell = 1$ so that
$U = [ u_1, u_2, u_3, u_4, u'_1 ]$ and
$V = [ v_1, v'_1, v'_2, v'_3, v'_4 ]$.
We also only assume that the functions $u_i, u'_i, v_i$ and $v'_i$ are defined only on the intervals in which they 
appear (e.g. $u_5$ is undefined, as is $v'_1$).
Then we have
\begin{equation*}
	U + V = (u_i + v_i)^{[\![1,4]\!]} \oplus (u'_{i-4} + v_i)^{(\!(4,1]\!]} \oplus (u'_{i-4} + v'_{i-1})^{(\!(1,5]\!]}.
\end{equation*}

None of the individual sub-terms cannot be evaluated directly.
In the first term, $v_i$ is not totally defined over the interval $[\![1,4]\!]$.
In the third term, on the interval $(\!(1,5]\!]$, $u'_{i-4}$ would even evaluated on negative indices.
However, these un-evaluable terms also appear in the middle term, however the interval $(\!(4,1]\!]$ 
has negative mutliplicity so the offending points end up with mutliplicity zero so are properly ignored.
\begin{align*}
	U + V
		&= (u_i + v_i)^{[\![1,1]\!] \oplus (\!(1,4]\!]} 
			\oplus (u'_{i-4} + v_i)^{\ominus(\!(1,4]\!]} 
			\oplus (u'_{i-4} + v'_{i-1})^{(\!(1,4]\!] \oplus (\!(4,5]\!]}\\
		&= (u_i + v_i)^{[\![1,1]\!]} 
			\oplus \left((u_i + v_i) - (u'_{i-4} + v_i) + (u'_{i-4} + v'_{i-1})\right)^{[\![1,4]\!]} 
   \\&
			\oplus (u'_{i-4} + v'_{i-1})^{(\!(4,5]\!]} \\ 
		&= (u_i + v_i)^{[\![1,1]\!]} 
			\oplus (u_i + v'_{i-1})^{(\!(1,4]\!]} 
			\oplus (u'_{i-4} + v'_{i-1})^{(\!(4,5]\!]}.
\end{align*}

%
%

\section{Higher Dimension Intervals}
\label{sec:HigherDimensionIntervals}

Hybrid intervals work perfectly well when dealing with the indices of a vector. 
However, we are more interested in the rectangular blocks of a matrix.
We can move from 1-dimensional intervals to 2-dimensional blocks using the Cartesian product
\begin{definition}
	Let $X = \hset{ x_1^{m_1}, ... , x_k^{m_k} }$ and $Y= \hset{ y_1^{n_1}, ... , y_\ell^{n_\ell} }$ be hybrid sets
	over sets $S$ and $T$
	We define the \textbf{Cartesian product of hybrid sets $\boldsymbol{X}$ and $\boldsymbol{Y}$}, to be a hybrid set
	over $S \times T$ and denoted with $\times$ operator as
	\begin{equation*}
		X \times Y = \hset{ (x, y)^{m \cdot n} \; : \; x \in^m X, y \in^n Y }.
	 \end{equation*}
\end{definition}

If $[\![a,b]\!]$ and $[\![c,d]\!]$ are both have positive mutliplicity in $\mathbb{R}$ then their Cartesian product 
$[\![a,b]\!] \times [\![c,d]\!]$, as shown in Figure~\ref{fig:productofintervals}, is clearly a two dimensional rectangle in $\mathbb{R}^2$.
If one of $[\![a,b]\!]$ or $[\![c,d]\!]$ had negative mutliplicity then we would have a hybrid rectangle with points of negative mutliplicity.
If both intervals had negative mutliplicity, then the signs would combine to give points in the Cartesian product with \emph{positive} mutliplicity.

\begin{figure}[t] 
	\centering 
	\begin{tikzpicture}[y=0.75cm, x=1.5cm]	
	\draw(0,0) -- coordinate (x axis mid) (4,0);
    	\draw (0,0) -- coordinate (y axis mid) (0,4);
    	
    	\draw[fill] (1,1pt) rectangle (3,-1pt);    	
    	\draw (1, 3pt) -- (1, -3pt) node[anchor=north] {$a$};
    	\draw (3, 3pt) -- (3, -3pt) node[anchor=north] {$b$};
    	\draw (2, 0) node[anchor=north] {$[\![a,b]\!]$};
    	
    	\draw[fill] (1pt,1) rectangle (-1pt, 3);
    	\draw (3pt, 1) -- (-3pt, 1) node[anchor=east] {$c$};
    	\draw (3pt, 3) -- (-3pt, 3) node[anchor=east] {$d$};
    	\draw (0,2) node[anchor=east] {$[\![c,d]\!]$};
    	
    	\draw[fill, color=black!20] (1,1) rectangle (3,3);
    	\draw (2, 2) node {$[\![a,b]\!] \times [\![c,d]\!]$ };
	\end{tikzpicture}
	\caption[Cartesian product of two 1-rectangles] { 
		The Cartesian product of two positive multiplicity 1-rectangles $[\![a,b]\!]$ and $[\![c,d]\!]$ 
		is a positive multiplicity 2-rectangle.
	\label{fig:productofintervals}}
\end{figure}
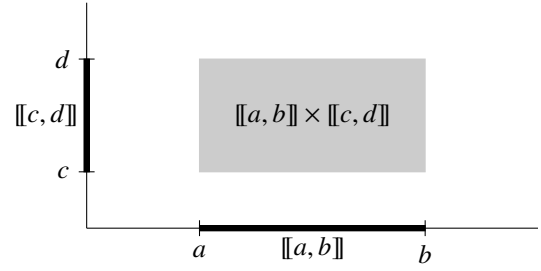

There is no reason to stop here.
$[\![a,b]\!]\times [\![c,d]\!]$ is still a hybrid set, we can take its Cartesian product with another interval, say $[\![e,f]\!]$
to get a rectangular cuboid in $\mathbb{R}^3$.
We should note here that we do not distinguish between $((x,y),z)$ and $(x,(y,z))$
 but rather we treat both as different names for the ordered triple $(x,y,z)$.
That is, the Cartesian product is associative,
\begin{gather*}
	\hset{ ((x, y), z)^{(m \cdot n)\cdot p} | x \in^m X, y \in^n Y, z \in^p Z }
	= X \times Y \times Z \\
	= \hset{ (x, (y, z)^{m \cdot (n\cdot p)} |  x \in^m X, y \in^n Y, z \in^p Z }.
\end{gather*}

Although we will not be using them in this article, 
the objects resulting from iterated Cartesian product of intervals 
turn out to be quite useful, and  we will call them hybrid $k$-rectangles. 
A non-degenerate hybrid interval is a hybrid 1-rectangle, a cross product of two is a hybrid 2-rectangle, and so on.
In the following, we use the shorthand $k$-rectangle to mean hybrid $k$-rectangle.

\begin{theorem}
	The Cartesian product of a $k$-rectangle in $\mathbb{R}^m$ (where, $k\leq m$) 
	and $\ell$-rectangle in $\mathbb{R}^n$ (again, $\ell \leq n$) 
	is a $(k+\ell)$-rectangle in $\mathbb{R}^{m+n}$.
\end{theorem}

For completeness we also define a 0-rectangle as a hybrid set containing a single point with multiplicity $1$ or $-1$.
This allows us to embed $k$-rectangles in $\mathbb{R}^n$.
For example $[\![a,b]\!]_\mathbb{R} \times [\![c,d]\!]_\mathbb{R} \times \hset{e^1}$ is the product of two 1-rectangles
and a 0-rectangle and so it is a 2-rectangle.
But it is still a Cartesian product of 3 hybrid sets (each over $\mathbb{R}$) and so is a 2-rectangle in $\mathbb{R}^3$.
Specifically, it is the 2-rectangle $[\![a,b]\!] \times [\![c,d]\!]$ on the plane $z=e$.
This also illustrates the principle that given a $k$-rectangle in $\mathbb{R}^n$ where $n>k$ we can always find a $k$ 
dimensional subspace which also contains the rectangle.

Finally, one last note regarding $k$-rectangles before we return to the topic of symbolic linear algebra.
We will re-use the interval notation and allow for hybrid intervals between two vectors: $[\![\boldsymbol{a}, \boldsymbol{b}]\!]$.
But one should be careful when interpreting this overloaded notation.
When $a$ and $b$ are real numbers, we continue to use the definition $[\![a,b]\!] = [a,b) \ominus [b,a)$.
However, when $\boldsymbol{a}$ and $\boldsymbol{b}$ are $n$-tuples (for example, coordinates in $\mathbb{R}^n$ 
then this is \emph{not} the hybrid line interval, 
$[\boldsymbol{a}, \boldsymbol{b}) \ominus [\boldsymbol{b}, \boldsymbol{a})$
rather we define it as follows:

\begin{definition}
	Let $\boldsymbol{a} = (a_1, a_2, \ldots, a_n)$ and 
	$\boldsymbol{b} = (b_1, b_2, \ldots, b_n)$ be ordered $n$-tuples then we use the notation:
	\begin{equation*}
		[\![ \boldsymbol{a}, \boldsymbol{b} ]\!] 
		= [\![a_1, b_1]\!] \times [\![a_2, b_2 ]\!] \times \ldots \times [\![a_n , b_n]\!].
	\end{equation*}
\end{definition}

The dimension of $[\![ \boldsymbol{a}, \boldsymbol{b} ]\!]$ is equal to the number of indices where $a_i$ and $b_i$ are distinct.
For any $i$ where $a_i = b_i$, the corresponding term: $[\![ a_i, b_i ]\!]$ will be a hybrid set containing a single point, that is, a 0-rectangle.
The multiplicity of $[\![ \boldsymbol{a}, \boldsymbol{b} ]\!]$ is based on the number of negative multiplicity intervals $[\![a_i,b_i]\!]$.
Should there be an odd number of indices $i$ such that $a_i > b_i$ then $[\![ \boldsymbol{a}, \boldsymbol{b} ]\!]$ will also have negative multiplicity.
Otherwise, it will have positive multiplicity.

For the remainder of this article, 
we will be interested only in matrices of dimension in 
$\mathbb{N}_0 \times \mathbb{N}_0$.
Here there is only room for a single Cartesian product and so this notation will not be immediately useful. 

\section{Matrix Addition}
\label{sec:MatrixAddition}

We now consider the addition of $2 \times 2$ block matrices $A$ and $B$ with overall dimensions $n \times m$ 
of the form
\begin{equation*}
	A = \left[ \begin{array}{c|c} A_{11} & A_{12} \\ \hline A_{21} & A_{22} \end{array} \right]
	\;\;\;\;\; \text{and} \;\;\;\;\;
	B = \left[ \begin{array}{c|c} B_{11} & B_{12} \\ \hline B_{21} & B_{22} \end{array} \right].
\end{equation*}
Since these are block matrices then $A_{ij}$ and $B_{ij}$ are not entries but sub matrices themselves.
We shall assume that $A_{11}$ is a $(q \times r)$ matrix and $B_{11}$ is a $(s \times t)$ matrix.
The sum of $A$ and $B$ will also be a $n \times m$ matrix.
Our universe, $\mathcal{U}$ is therefore the the set of all indices in an $n \times m$ matrix:
\begin{align*}
	\mathcal{U} 
		&=\; [\![0,n)\!)_{\mathbb{N}_0} \times [\![0,m)\!)_{\mathbb{N}_0}  \\
		&=\; \{ (i,j) \;|\; 0 \leq i < n \text{ and } 0 \leq j < m \text{ and } i,j \in \mathbb{N}_0 \} .
\end{align*}

First we must convert $A$ and $B$ to hybrid function notation. 
We use $\mathcal{A}_{ij}$ an $\mathcal{B}_{ij}$ to respectively denote the regions for which 
$A_{11}$ and $B_{ij}$ are defined.
Explicitly, these are
\begin{align*}
	\mathcal{A}_{11} &= [\![0,q)\!) \times [\![0,r)\!) &
	\mathcal{A}_{12} &= [\![0,q)\!) \times [\![r,m)\!) \\
	\mathcal{A}_{21} &= [\![q,n)\!) \times [\![0,r)\!) &
	\mathcal{A}_{22} &= [\![q,n)\!) \times [\![r,m)\!) \\
	\mathcal{B}_{11} &= [\![0,s)\!) \times [\![0,t)\!) &
	\mathcal{B}_{12} &= [\![0,s)\!) \times [\![t,m)\!) \\
	\mathcal{B}_{21} &= [\![s,n)\!) \times [\![0,t)\!) &
	\mathcal{B}_{22} &= [\![s,n)\!) \times [\![t,m)\!),
\end{align*}
which allow us to rewrite $A$ and $B$ as
\begin{align*}
	A &= A_{11}^{\mathcal{A}_{11}} \oplus 
		A_{12}^{\mathcal{A}_{12}} \oplus 
		A_{21}^{\mathcal{A}_{21}} \oplus 
		A_{22}^{\mathcal{A}_{22}} \\
	B &= B_{11}^{\mathcal{B}_{11}} \oplus 
		B_{12}^{\mathcal{B}_{12}} \oplus 
		B_{21}^{\mathcal{B}_{21}} \oplus 
		B_{22}^{\mathcal{B}_{22}}.
\end{align*}

Depending on the relation of $q$ with $s$ and $r$ with $t$ the regions in the sum of $A$ and $B$ may vary.
The shapes of block matrices that can arise have shown
in Figure~\ref{fig:MatAdditionPermutations}. 
Intuitively, the approach we take is to not concern ourselves with all possible cases that \emph{could} arise but to just choose one ordering.
\textbf{If this ordering is wrong, then the hybrid function multiplicities will cancel to yield the correct expression regardless.}

Since there are 4 partitions in $A$ and 4 partitions in $B$, we only require 7 pieces to form a common refinement.
To this refinement for, we follow the same method as used previously:
\begin{equation}
	\label{eqn:2x2CommonRefinement}
	\Big\{ \;
		\mathcal{A}_{11}, \; \mathcal{A}_{12}, \;  \mathcal{A}_{21}, \;
		\mathcal{B}_{11}, \; \mathcal{B}_{12}, \; \mathcal{B}_{21}, \; \mathcal{P} 
	\; \Big\}
\end{equation}
with $\mathcal{P}$ defined as
\begin{equation*}
	\mathcal{P} = \mathcal{U} 
		\ominus \left( \mathcal{A}_{11} \oplus \mathcal{A}_{12} \oplus \mathcal{A}_{21} \oplus
				\mathcal{B}_{11} \oplus \mathcal{B}_{12} \oplus \mathcal{B}_{21} \right).
\end{equation*}
Clearly we can still express $\mathcal{A}_{22}$ using only the terms from the common refinement by
\begin{align*}
	\mathcal{A}_{22} 
		&= \mathcal{U} \; \ominus \; (\mathcal{A}_{11} \oplus \mathcal{A}_{12} \oplus \mathcal{A}_{21}) \\[-0.1em]
		& = \mathcal{U} \; \ominus \;
            (\mathcal{A}_{11} \oplus \mathcal{A}_{12} \oplus \mathcal{A}_{21}  
			\oplus \mathcal{B}_{11} \oplus \mathcal{B}_{12} \oplus \mathcal{B}_{21})  \\[-0.6ex]&\phantom{=}\;
			\oplus \mathcal{B}_{11} \oplus \mathcal{B}_{12} \oplus \mathcal{B}_{21}\\[-0.1em]
		&= \mathcal{P} \oplus \mathcal{B}_{11} \oplus \mathcal{B}_{12} \oplus \mathcal{B}_{21}.
\end{align*}
Similarly $\mathcal{B}_{22}$ can be represented as $\mathcal{B}_{22} = \mathcal{P} \oplus \mathcal{A}_{11} \oplus \mathcal{A}_{12} \oplus \mathcal{A}_{21}$
and $\mathcal{U}$ as the sum of all 7 regions, 
$\mathcal{U} = 	\mathcal{A}_{11} \oplus \mathcal{A}_{12} \oplus \mathcal{A}_{21} \oplus 
				\mathcal{B}_{11} \oplus \mathcal{B}_{12} \oplus \mathcal{B}_{21} \oplus \mathcal{P}$.
Thus $A$ and $B$ can be rewritten using this new generalized partition as
\begin{align*}
	A &= A_{11}^{\mathcal{A}_{11}} \oplus 
		A_{12}^{\mathcal{A}_{12}} \oplus 
		A_{21}^{\mathcal{A}_{21}} \oplus 
		A_{22}^{\mathcal{P} \oplus \mathcal{B}_{11} \oplus \mathcal{B}_{12} \oplus \mathcal{B}_{21}} \\
	B &= B_{11}^{\mathcal{B}_{11}} \oplus 
		B_{12}^{\mathcal{B}_{12}} \oplus 
		B_{21}^{\mathcal{B}_{21}} \oplus 
		B_{22}^{\mathcal{P} \oplus \mathcal{A}_{11} \oplus \mathcal{A}_{12} \oplus \mathcal{A}_{21}} .
\end{align*}

With this, addition becomes straightforward:
we add functions for terms over corresponding regions.
Since we are using \emph{generalized partitions}, not traditional partitions we cannot guarantee disjointness.
We must also apply a \mbox{$+$-reduction} after summing each matching pair:
\begin{align*}
	(A+B) = \R[+] & \left(  (A_{11}+B_{22})^{\mathcal{A}_{11}} \oplus 
		(A_{12} + B_{22})^{\mathcal{A}_{12}} \oplus 
		(A_{21} + B_{22})^{\mathcal{A}_{21}} \right. \\[-0.5em] &\oplus 
		(A_{22} + B_{11})^{\mathcal{B}_{11}} \oplus 
		(A_{22} + B_{12})^{\mathcal{B}_{12}} \oplus 
		(A_{22} + B_{21})^{\mathcal{B}_{21}} \\[-0.5em] &\oplus
		\left.(A_{22} + B_{22})^{\mathcal{P}} \right).
\end{align*}

\subsection{Example: \emph{Evaluation at points}} 
We now demonstrate evaluating this expression.
Let us assume a point $(i,j)$ exists in the region $\mathcal{A}_{11} \cap \mathcal{B}_{12}$.
That is, $0 \leq i < \min(q,s)$ and $t \leq j < r$. 
Evaluating each of the hybrid sets from \eqref{eqn:2x2CommonRefinement} we find that only three have 
non-zero multiplicities: $\mathcal{A}_{11}(i,j)=1$, $\mathcal{B}_{12}=1$ and 
$\mathcal{P}(i,j)=1-(1+0+0+0+1+0)=-1$.
After removing all zero terms, this yields:
\begin{align*}
	(A+B)(i,j) &= \R[+]  \left(  (A_{11}+B_{22})^{1} \oplus 
		(A_{22} + B_{12})^{1} \oplus 
		(A_{22} + B_{22})^{-1} \right)\\
		&= \left(A_{11}+B_{22}) + (A_{22} + B_{12}) - (A_{22} + B_{22}\right)(i,j)\\
		&= (A_{11}+B_{12})(i,j).
\end{align*}

As a second example assume $(i,j) \in \mathcal{A}_{22} \cap \mathcal{B}_{12}$.
Then we find there is only one partition with non-zero multiplicity.
Clearly $\mathcal{B}_{12} = 1$ but $\mathcal{A}_{22} \notin$\eqref{eqn:2x2CommonRefinement}.
Calculating the multiplicity of $\mathcal{P}$ also yields $1-(0+0+0+0+1+0) = 0$.
Very simply:
\begin{align*}
	(A+B)(i,j) &= \R[+]  \left(  (A_{22}+B_{12})^{1} \right)(i,j)\\
		&=(A_{22}+B_{12})(i,j).
\end{align*}

\subsection{Addition with Larger Block Matrices}
This method extends easily from addition of two $2\times 2$ block matrices to arbitrary addition of block matrices.
If we consider (\emph{conformable}) $k \times \ell$ and $n \times m$ block matrices $A$ and $B$ respectively of the form:

\begin{equation*}
	A = \begin{bmatrix}
		A_{11} & \ldots & A_{1\ell}\\
		\vdots & & \vdots \\
		A_{k1} & \ldots & A_{k\ell}	
	\end{bmatrix}
	\;\;\;\;\;
	\text{ and }
	\;\;\;\;\;
	B = \begin{bmatrix}
		B_{11} & \ldots & B_{1m}\\
		\vdots & & \vdots \\
		B_{n1} & \ldots & B_{nm}	
	\end{bmatrix}.
\end{equation*}

For matrices to be conformable for addition they must have the same dimensions.
So we can partition the rows of $A$ by the strictly increasing sequence $\{q_i\}_{i=0}^k$ and 
the columns by $\{r_j \}_{j=0}^\ell$.
Similarly for $B$ we partition the rows by $\{s_i\}_{i=0}^n$ and the columns by $\{t_j\}_{j=0}^m$.
With the additional constraints that ${q_0 = r_0 = s_0 = t_0 = 0}$ and $q_k = s_n$ and $r_\ell = t_m$.
Each $A_{ij}$ and $B_{ij}$ is defined over a rectangular region $\mathcal{A}_{ij}$ and $\mathcal{B}_{ij}$:
 \begin{equation*}
	\mathcal{A}_{ij} = [\![q_{i-1}, q_i )\!) \times [\![ r_{j-1}, r_{j} )\!)
	\;\;\;\;\;
	\mathcal{B}_{ij} = [\![s_{i-1}, s_i )\!) \times [\![ t_{j-1}, t_{j} )\!).
\end{equation*}
which gives the expression
\begin{equation*}
\begin{aligned}
	(A+B) = \R[+] 
		&\left(\; \left( \bigoplus_{(i,j) \neq (n,m)} (A_{ij} + B_{nm})^{\mathcal{A}_{ij}} \right) \right .\oplus  \nonumber \\
		&\;\; \left .\left( \bigoplus_{(i,j) \neq (n,m)} (A_{nm} + B_{ij})^{\mathcal{B}_{ij}} \right) \oplus 
		 	(A_{nm} + B_{nm})^{\mathcal{P}} 
	\right).
\end{aligned}
\end{equation*}

%
\section{Matrix Multiplication}
\label{sec:MatrixMultiplication}

Next we consider the product of symbolic block matrices.
Again, we will assume $2 \times 2$ block matrices $A$ and $B$.
For these matrices to be conformable for multiplication they must be  $n \times m$ and $m \times p$ rather than
the same size as is required for addition,
\begin{equation*}
	A = \begin{bmatrix} A_{11} & A_{12} \\ A_{21} & A_{22} \end{bmatrix}
	\;\;\;\;\; \text{and} \;\;\;\;\;
	B = \begin{bmatrix} B_{11} & B_{12} \\ B_{21} & B_{22} \end{bmatrix}.
\end{equation*}
Where $A_{11}$ is a $q \times r$ matrix and $B_{11}$ is a $s \times t$ matrix.
Note that $0 \leq r , s \leq m$ but the ordering of $r$ and $s$ is unknown.

In the simplest case, $r=s$, four regions will arise each with simple closed expressions,
\begin{equation*}
	AB = \begin{bmatrix}
		\left( A_{11}B_{11}+A_{12}B_{21} \right) & \left( A_{11}B_{12}+A_{12}B_{22} \right) \\ 
		\left( A_{21}B_{11}+A_{22}B_{21} \right) & \left( A_{21}B_{12}+A_{22}B_{22} \right)
	\end{bmatrix}.
\end{equation*}
One should notice the similarity between this and multiplication of simple $2 \times 2$ matrices.
If we consider only the top-left block, since $r=s$ then the $(q \times r)$ matrix $A_{11}$ and the 
$(s \times t)$ matrix $B_{11}$ are conformable.
As are the $(q \times m-r)$ matrix $A_{12}$ and the $(m-s \times t)$ matrix $B_{21}$.
Both products will result in a $q \times t$ matrix which are conformable for addition.
Thus the term $A_{11}B_{11} + A_{12}B_{21}$ is a $q \times t$ block.

If $r \neq s$ then one approach would be to partition $A$ into a $2 \times 3$ block matrix 
split along the vertical lines $r$ and $s$ and the horizontal line $q$.
And split $B$ into a $3 \times 2$ block matrix split along the vertical line $t$  and the horizontal lines $r$ and $s$:
Depending on the relative ordering of $r$ and $s$ this may cause different blocks to be split.
If $s < r$ then $A_{11}$ and $A_{21}$ will be split into blocks with columns from 0 to $s$ and then from $s$ to $r$
while $B_{21}$ and $B_{22}$ would be split into blocks with rows from $s$ to $r$ and from $r$ to $m$,
\begin{equation*}
	A= \left[ \begin{array}{cc|c}
			A_{11}^{(1)} & A_{11}^{(2)} & A_{12}^{} \\[3pt]
			\hline
			A_{21}^{(1)} & A_{21}^{(2)} & A_{22}^{}  \rule{0pt}{12pt}
		\end{array} \right]
	\;\;\;\;\;\text{and}\;\;\;\;\;
	B = \left[ \begin{array}{c|c}
			B_{11}^{} & B_{12}^{} \\[3pt]
			\hline
			B_{21}^{(1)} & B_{22}^{(1)} \rule{0pt}{12pt}\\[3pt]
			B_{21}^{(2)} & B_{22}^{(2)} 
		\end{array} \right].
\end{equation*}

The resulting product is still a $2 \times 2$ matrix.
Additionally, each block is still the same size; the first block in the top-left is still $q \times t$.
However each block is now the sum of three block products:
\begin{align*}
	&AB 	= 	\\ &\begin{bmatrix}
				\left( A_{11}^{(1)}B_{11}^{}+ A_{11}^{(2)}B_{21}^{(1)} + A_{12}^{}B_{21}^{(2)} \right) & 
				\left( A_{11}^{(1)}B_{12}^{}+ A_{11}^{(2)}B_{22}^{(1)} + A_{12}^{}B_{22}^{(2)} \right) \\[5pt]
				\left( A_{21}^{(1)}B_{11}^{}+ A_{21}^{(2)}B_{21}^{(1)} + A_{22}^{}B_{21}^{(2)} \right) & 
				\left( A_{21}^{(1)}B_{12}^{}+ A_{21}^{(2)}B_{22}^{(1)} + A_{22}^{}B_{22}^{(2)} \right) 
			\end{bmatrix}.
\end{align*}
On the other hand, if $r < s$ then $A_{12}$ and $A_{22}$ will be the blocks split vertically while $B_{11}$ and $B_{12}$
will be split horizontally. 
In turn, this leads to a different expression for the product of $A$ and $B$.
In a now familiar pattern we can use hybrid functions to give a single expression 
to deal with all permutations simultaneously.

We shall refer to the product $AB$ as the block matrix $C$,
\begin{equation*}
	C = AB = \begin{bmatrix} C_{11} & C_{12} \\ C_{21} & C_{22} \end{bmatrix}.
\end{equation*}
$C$ is an $n \times p$ matrix as determined by the sizes of $A$ and $B$ and $C_{11}$ is a $q \times t$ sub-matrix.
This leaves $C_{12}$, $C_{21}$ and $C_{22}$ to be $q \times (p-t)$, $(n-q) \times t$ and $(n-q) \times (p-t)$ respectively.
We partition all three matrices along the axes $0.. n$, $0..p$ and $0..m$ into the hybrid intervals
\begin{align*}
	N_1 	&= [\![0, q)\!) 	& N_2 	&= [\![q, n)\!) 	\\
	P_1 	&= [\![0, t)\!) 	& P_2 	&= [\![t, p)\!) 	\\
	M_1 	&= [\![0, r)\!) 	& M_2 	&= [\![r, s)\!) 	& M_3 	&= [\![s, m)\!).
\end{align*}

Assumption is too strong a word, but these partitions follow the \emph{guess} that $r<s$.
So we construct expressions with this in mind. 
If we chose incorrectly, then we plan to use the negative multiplicity of $M_2$ to correct our expression.
Using these intervals, we can now rewrite our matrices inline as
\begin{equation}
\begin{aligned}
	A & =	A_{11}^{N_1 \times M_1} \oplus A_{12}^{N_1 \times (M_2 \oplus M_3)} \oplus 
			A_{21}^{N_2 \times M_1} \oplus A_{22}^{N_2 \times (M_2 \oplus M_3)} \\
	B & =	B_{11}^{(M_1 \oplus M_2) \times P_1} \oplus B_{12}^{(M_1 \oplus M_2) \times P_2} \oplus 
			B_{21}^{M_3 \times P_1} \oplus B_{12}^{M_3 \times P_2}\\
	C & =	C_{11}^{N_1 \times P_1} \oplus C_{12}^{N_1 \times P_2} \oplus
			C_{21}^{N_2 \times P_1} \oplus C_{22}^{N_2 \times P_2}.
\end{aligned}
	\label{eqn:2x2multiplicationblocks}
\end{equation}
It should be noted here that $\oplus$ is still the point-wise sum of hybrid functions.
It should not be confused with the direct sum or the Kronecker sum of matrices which both use the same $\oplus$ operator.
The $\times$ operator refers to the Cartesian product of intervals.

For $i,j \in \{ 1,2 \}$ the terms of $C$ are given by
\begin{equation}
\begin{aligned}
	C_{i,j}^{N_i \times P_j} (x,y) = \sum_{M} \R[\smallmult]  
		&\left( \;\;
			\left. 	A_{i,1}^{N_1 \times M_1}	\right|_{X=x} \;\oplus\;
			\left.	B_{1,j}^{M_1 \times P_1}	\right|_{Y=y} 
		\right.  \\
	 	&\oplus
	 		\left.	A_{i,2}^{N_1 \times M_2}	\right|_{X=x} \;\oplus\;
			\left. 	B_{1,j}^{M_2 \times P_1}	\right|_{Y=y} 
		\\
		&\oplus\left.
			\left.	A_{i,2}^{N_1 \times M_3}	\right|_{X=x} \;\oplus\;
			\left. 	B_{2,j}^{M_3 \times P_1}	\right|_{Y=y}
		\;\;\right).
\end{aligned}
	\label{eqn:2x2multiplication}
\end{equation}

There is some new notation here so let us unpack it.
Recall that we are taking the approach that matrices are simply functions defined on $\mathbb N_0 \times \mathbb N_0$.
As a function we can take a restriction of a matrix to a set of indices.
In the above, we use $X$ and $Y$ to denote the row and column indexing respectively.
For example with the matrix $M$, given below $M|_{X=0}$ and $M_{Y=0}$ would be as follows:
\begin{gather*}
	M = \begin{bmatrix}
		M[0,0] 	& \ldots 	& M[0,n] \\
		\vdots 	& 		& \vdots \\
		M[m,0]	& \ldots & M[m,n]
	\end{bmatrix} \\
	M|_{X=0} = \begin{bmatrix}
		M[0,0] & \ldots & M[0,n]
	\end{bmatrix} \quad
	M|_{Y=0} = \begin{bmatrix}
		M[0,0] \\ \vdots \\ M[m,0]
	\end{bmatrix}.
\end{gather*}

This is more powerful than just simple evaluation.
We are selecting not a fixed axis as $(x,y)$ is the input to our function.
So for a matrix $M|_{X=x}$ or $M|_{Y=y}$ we transform $M:X\times Y \to Z$
to the curried $M|_{X=i}:Y \to ( X \to Z)$ or $M|_{Y=j}:X \to (Y \to Z)$.
Within the context of equation \eqref{eqn:2x2multiplication}, this transforms the blocks of $A$ into horizontal vector slices 
and $B$ into vertical slices.

Ignoring the differences in transposition, when thought of as functions, these both map from $M$ 
(the common axis of $A$ and $B$) to functions with a common range.
We therefore have the pointwise sum of terms of the forms $m \mapsto ( x \mapsto A[x][m] )$ and 
$m \mapsto (y \mapsto B[m][y])$.
The work of multiplying matching $A[x][m]$ with $B[m][y]$ is handled by the $\R[\smallmult]$.
This leaves us with the product of two functions with different domains, but common range:
\begin{equation*}
	 ( x \mapsto A[x][m] ) \mult (y \mapsto B[m][y]) 
	 = (x,y) \mapsto A[x][m] \mult B[m][y].
\end{equation*}

Finally, we have the sum over $M$.
If $A$ and $B$ are matrices over a field $F$ then the \mbox{$\times$-reduction} 
yields a function $M \to (N \times P \to F)$.
Summing over the set $M$ leaves us with a function $(N \times P \to F)$ which agrees (at least by object type) 
with our expectations for $C$.
The familiar structure of summing over a product suggest correctness when $\big\{ M_1, M_2, M_3 \big\}$ 
is a strict partition of $M$ (that is, when $r \leq s$).
Despite the mental hurdles of say a $2 \times (-3)$ matrix, it continues to hold for general partitions as well.

\subsection{Example: \emph{Matrix Multiplication Concretely}}

We consider the product of two block matrices $Q$ and $R$.
For this example, to better differentiate between blocks, 
we change our notation slightly and give each block a distinct letter name: 
$A,B,C,D$ for the blocks of $Q$ and $E,F,G,H$ for the blocks of $R$,
\begin{equation*}
	Q = \left[ \begin{array}{cc|c}
		a_1 & a_2 & b_1 \\
		a_3 & a_4 & b_2 \\
		\hline
		c_1 & c_2 & d_1 \\
		c_3 & c_4 & d_2
	\end{array} \right]
	\;\;\; \text{and} \;\;\;
	R = \left[ \begin{array}{c|cccc}
		e_1 & f_1 & f_2 & f_3 & f_4 \\
		\hline
		g_1 & h_1 & h_2 & h_3 & h_4 \\
		g_2 & h_5 & h_6 & h_7 & h_8
	\end{array} \right].
\end{equation*}

We again use $N$, $M$ and $P$ for the sets of indices.
As $4 \times 3$ and $3\times 5$ matrices, we have $N = [\![0,3]\!]$, $M=[\![0,2]\!]$ and $P=[\![0,4]\!]$.
To align with the blocks of $Q$ and $R$, each of these sets is partitioned as follows:
\begin{align*}
	N_1 	&= [\![0, 1]\!] 	& N_2 	&= [\![2, 3]\!] 	\\
	P_1 	&= [\![0]\!] = \hset{0^{+1}}		& P_2 	&= [\![1, 4]\!] 	\\
	M_1 	&= [\![0, 1]\!] 	& M_2 	&= (\!(1)\!) = \hset{1^{-1}} 	& M_3 	&= [\![1, 2]\!].
\end{align*}
We should note here that our guess was wrong: $M_2$ has negative multiplicity!
Although we could have constructed two expressions to handle this case as well, this is not necessary.
We can continue as if nothing is wrong, and the hybrid function structure cancels multiplicities where needed.

We can still write $Q$  and $R$ as
\begin{align*}
	Q &= A^{N_1 \times M_1} \oplus 
		B^{N_1 \times (M_2 \oplus M_3)} \oplus
		C^{N_2 \times M_1} \oplus
		D^{N_2 \times (M_2 \oplus M_3)}\\
	R &= E^{(M_1 \oplus M_2) \times P_1} \oplus
		F^{(M_1 \oplus M_2) \times P_2} \oplus
		G^{M_3 \times P_1} \oplus
		H^{M_3 \times P_2}	.
\end{align*}
The only difference is that originally the sum $(M_2 \oplus M_3) = \{ 2 \}$ was intended to \emph{extend} $M_3$.
When $M_2$ is negative, it is a set of indices which is \emph{smaller} than the $M_3 = \{ 1,2 \}$ we started with.
Similarly, in the expression for $R$, $(M_1 \oplus M_2)$ is smaller than $M_1$.
We use $S$ to denote the product $QR$ which is still another $2\times 2$ block matrix by the same construction 
as equation \eqref{eqn:2x2multiplicationblocks}:
\begin{equation*}
	S = Q \cdot R  =
		\begin{bmatrix} S_1 & S_2 \\ S_3 & S_4 \end{bmatrix} =
		 {S_1}^{N_1 \times P_1} \oplus
		 {S_2}^{N_1 \times P_2} \oplus
		 {S_3}^{N_2 \times P_1} \oplus
		 {S_4}^{N_2 \times P_2}.
\end{equation*}

To continue the example, we compute the block $S_1$,
\spacetune{\pagebreak}
\begin{align*}
	{S_1}^{N_1 \times P_1}(i,j) &= \sum_{m \in M} \R[\smallmult]  \left( 
			\left. A^{N_1 \times M_1}\right|_{X=i} \oplus
			\left. E^{M_1 \times P_1}\right|_{Y=j} \oplus \right.\\
			&\;\left. B^{N_1 \times M_2}\right|_{X=i} \oplus
			\left. E^{M_2 \times P_1}\right|_{Y=j} \oplus\\ 
			&\;\left. \left. B^{N_1 \times M_3}\right|_{X=i} \oplus
			\left. G^{M_3 \times P_1}\right|_{Y=j}
	\right).
\end{align*}
As this is a small example, our curried functions only range over $\{ 0, 1, 2 \}$.
This is a small enough domain to express each of the functions as a set of point-wise mappings.
We therefore expand out each of our terms as formal \emph{hybrid sets} 
(recall a hybrid function is a special hybrid set of ordered pairs):
\begin{align*}
	\sum_{m \in M ?} \R[\smallmult] \bigg( 
		&\;\;\hset{
			\left( 0 \mapsto 
				\left[\begin{smallmatrix}\vphantom{b}a_1 \\ \vphantom{b}a_3 \end{smallmatrix}\right]
			\right)^{+1}, \;
			\left( 1 \mapsto 
				\left[\begin{smallmatrix}\vphantom{b}a_2 \\ \vphantom{b}a_4 \end{smallmatrix}\right]
			\right)^{+1}}  \\
		&\oplus \hset{ 
			\left( 0 \mapsto [ e_1 ] \right)^{+1}, \;
			\left( 1 \mapsto [e_\bot] \right)^{+1}} \\ 
		&\oplus \hset{ 
			\left( 1 \mapsto \left[\begin{smallmatrix} b_\bot \\ b_\bot \end{smallmatrix}\right] \right)^{-1} } 
              \oplus \hset{ 
			\left( 1 \mapsto [ e_\bot ] \right)^{-1} } \\
		&\oplus \hset{
			\left( 1 \mapsto \left[ \begin{smallmatrix} b_\bot \\ b_\bot \end{smallmatrix} \right] \right)^{+1}, \;
			\left( 2 \mapsto \left[ \begin{smallmatrix} b_1 \\ b_2 \end{smallmatrix}\right] \right)^{+1} } \\
		&\oplus \hset{
			\left( 2 \mapsto [ g_1 ] \right)^{+1},\; 
			\left( 2 \mapsto [ g_2 ] \right)^{+1}}	\bigg).
\end{align*}
\spacetune{\balance}

We are using $e_\bot$ and $b_\bot$ here to represent that the functions $E$ and $B$ are undefined for these points.
In reality, we would simply not even attempt to evaluate $B|_{X=x}(1)$ or $E|_{Y=y}(1)$ as the functions are undefined.
These points are actually contained in the $A$ and $G$ blocks, once again we must delay evaluation with pseudo-functions.

Applying the $\times$-reduction $\R[\smallmult]$, we group terms by their input value (e.g. $1 \mapsto x$ with $1 \mapsto y$) 
and flatten using the multiplicity to repeat or invert the $\mult$ operator.
In this case, we are dealing only with multiplicities of $+1$ and $-1$ which correspond with multiplication and ``division''.
This is not true division, as $0 \mult^{-1} 0 = 1$ without fear of division by zero.
Otherwise for non-zero operands, $\mult^{-1}$ agrees with the normal understanding of division.
This is made possible by working with multiplication as a \emph{group} rather than as a \emph{ring}
and so we are not worried the interactions between multiplication and addition.
This yields
\begin{align*}
	\sum_{M_1 \oplus M_2 \oplus M_3}
		& \left\{\!\left|\; \left(0 \mapsto 
			\left[\begin{smallmatrix}\vphantom{b}a_1 \\ \vphantom{b}a_3 \end{smallmatrix}\right] \mult^{+1} 
			[e_1] \right), \right.\right.\\
		&\;\;\;\left(1 \mapsto 
			\left[\begin{smallmatrix}\vphantom{b}a_2 \\ \vphantom{b}a_4 \end{smallmatrix}\right] \mult^{+1}
			[e_\bot] \mult^{-1}
			\left[\begin{smallmatrix} b_\bot \\ b_\bot \end{smallmatrix}\right] \mult^{-1}
			[ e_\bot ]\mult^{+1}
			\left[ \begin{smallmatrix} b_\bot \\ b_\bot \end{smallmatrix} \right] \right)
			[ g_1 ],\\
		&\;\;\left.\left.\left(2 \mapsto 
			\left[ \begin{smallmatrix} b_1 \\ b_2 \end{smallmatrix}\right]\mult^{+1}
			[ g_2 ] \right) \; \right|\!\right\}.
\end{align*}
After some cancellations in the second term, we evaluate $\mult^{+1}$ as matrix multiplication and sum over all of $M$
\begin{equation*}
	{S_1}^{N_1 \times P_1} =
			\begin{bmatrix}\vphantom{b}a_1 \\ \vphantom{b}a_3 \end{bmatrix}
			[e_1]
		+ 	\begin{bmatrix}\vphantom{b}a_2 \\ \vphantom{b}a_4 \end{bmatrix}
			[g_1]
		+	\begin{bmatrix} b_1 \\ b_2 \end{bmatrix}
			[ g_2 ]
		=  \begin{bmatrix}a_1e_1+a_2g_1+b_2g_2\\ a_3e_1+a_4b_1+b_2g_2\end{bmatrix}.
\end{equation*}
As expected, we have a $|N_1| \times |P_1| = (2 \times 1)$ matrix which will form the upper left block of $S$.
Ignoring the block structure of $Q$ and $R$ and performing normal matrix multiplication, we also find that these values 
agree with $S[0,0]$ and $S[1,0]$.
Computations for the blocks $S_2$, $S_3$ and $S_4$ are performed similarly yielding blocks of varying sizes.
Together, these blocks form a strict partitioning of $S$ as a $2\times 2$ block matrix.

\subsection{Multiplication with Larger Block Matrices}

Extending to larger block matrices is a fairly trivial affair.
Once again we use $\{N_i\}$ to divide the rows of blocks in $A$ 
and $\{P_j\}$ to divide the block columns of $B$.
$M_k$ and $M'_k$ are two different partitions of the common axis for $A$ and $B$ respectively,

\begin{align*}
	A &= \left[ \begin{array}{c|c|c}
			A_{1,1}^{N_1 \times M_1} 	& \ldots 		& A_{1,K}^{N_1 \times M_K} \\
			\hline
			\vdots 						& 			& \vdots \\
			\hline
			A_{I,1}^{N_I \times M_1} 	& \ldots 		& A_{I,K}^{N_I \times M_K}
		\end{array}\right]
		&
	B &= \left[ \begin{array}{c|c|c}
			B_{1,1}^{M_1 \times P_1} 	& \ldots 		& B_{1,J}^{M_1 \times P_J} \\
			\hline
			\vdots 						& 			& \vdots \\
			\hline
			B_{K',1}^{M_{K'} \times P_1} & \ldots 		& B_{K',J}^{M_{K'} \times P_J}
		\end{array}\right]
	\\[8pt] 
	&= \bigoplus_{i \in [\![1,I]\!]} \bigoplus_{k \in [\![1,K]\!]} A_{i,k}^{N_i \times M_k}&
	&= \bigoplus_{k' \in [\![1,K']\!]} \bigoplus_{j \in [\![1,J]\!]} B_{k',j}^{M'_{k'} \times P_j}.
\end{align*}
As before, the blocks of $C$ will be of sizes $N_i \times P_j$:
\begin{equation*}
	C 
	\;=\; 
		\bigoplus_{i \in [\![1,I]\!]} \bigoplus_{j \in [\![1,J]\!]} \left( C_{i,j}^{N_i \times P_j} \right)
	\;=\; 	
		\left[ \begin{array}{c|c|c}
			C_{1,1}^{N_1 \times P_1} 	& \ldots 		& C_{1,J}^{N_1 \times P_J} \\
			\hline
			\vdots 						& 			& \vdots \\
			\hline
			C_{I,1}^{N_I \times P_1} 	& \ldots 		& C_{I,J}^{N_I \times P_J}
		\end{array}\right],
\end{equation*}
where each $C_{i,j}$ is defined as
\begin{equation*}
	C_{i,j} = 
	\sum_M 
		\R[\smallmult] \left(
			\bigoplus_{k \in [\![1,K]\!]} \restrict{ A_{i,k}^{N_i \times M_k} }{X=x} \oplus
			\bigoplus_{k \in [\![1,K']\!]} \restrict{ B_{k',j}^{M'_k \times P_j} }{Y=y} 
	 	\right) .
\end{equation*}

%
%
\section{Conclusions}
We have developed the concept of \emph{hybrid intervals} as intervals on ordered sets with multiplicities that can be negative.
This allows various identities to be written compactly, without having to consider multiple cases.

We have recast \emph{symbolic block matrix algebra} by viewing matrices as piecewise functions of hybrid intervals of indices.
This formulation allows addition and multiplication of matrices of symbolic block sizes without worrying about how the block sizes relate to each other:   Any generic relationship among the block sizes may be assumed and the 
hybrid interval
multiplicities assemble the correct arguments for any particular values of the symbolic block size parameters.
\label{sec:Conclusions}

\end{document}